\documentclass[a4paper, USenglish]{lipics-v2021}
\nolinenumbers
\usepackage{graphicx} 
\usepackage{amsthm}
\usepackage{amsmath}
\usepackage{amssymb}
\usepackage[dvipsnames]{xcolor}

\title{Substep Linearizability: Relaxing Atomicity While Preserving Sequential Specification}

\author{Finn Hackett\footnote{\label{equalcontrib}Both authors contributed to this work equally.}}
{University of British Columbia, Canada}
{}{}{}

\author{Evan Wrench\footnotemark[\value{footnote}]}
{University of British Columbia, Canada}
{}{}{}

\author{Ivan Beschastnikh}
{University of British Columbia, Canada}
{}{}{}

\author{Yuanhao Wei}
{University of British Columbia, Canada}
{}{}{}

\authorrunning{F. Hackett, E. Wrench, I. Beschastnikh, and Y. Wei}

\Copyright{Finn Hackett, Evan Wrench, Ivan Beschastnikh, Yuanhao Wei}

\ccsdesc{Theory of computation~Program specifications}
\ccsdesc{Theory of computation~Concurrency}

\keywords{Concurrency, Formal Specifications, Linearizability}

\relatedversion{}

\begin{document}

\maketitle

\begin{abstract}
    Linearizability is a standard correctness condition for concurrent objects.
    It takes a correct by construction sequential specification of a concurrent data structure, and produces a composable specification under atomic interleavings of operations.
    However, some operations like atomic exchangers are not specifiable atomically, making it difficult to reason compositionally about their specifications.
    Interval linearizability is the state of the art in relaxing linearizability's constraints.
    Formally it is a powerful framework, but it requires an interval-sequential specification of the target data structure, which we show is challenging to correctly generate or understand, even for simple operations on a single concurrent register.
    We propose substep linearizability, which we show is equivalent to interval linearizability using X-CCorrectness, a relaxed concurrency framework we created to compare existing linearizability semantics.
    We show that substep linearizability's sequential specifications are suitable for previously infeasible specification tasks, while being comparable in complexity to linearizability's sequential specifications.
    For evaluation, we model an atomic exchanger, a popular non-linearizable queue structure, and snapshot isolation between database transactions.
\end{abstract}



\def\H{\mathcal{H}}
\def\S{\mathcal{S}}

\newcommand{\hao}[1] {\textcolor{ForestGreen}{\bf (Hao)~#1}{\typeout{#1}}}
\newcommand{\hedit}[1] {\textcolor{ForestGreen}{{#1}}}
\newcommand{\evan}[1] {\textcolor{RedOrange}{\bf (Evan)~#1}{\typeout{#1}}}
\newcommand{\eedit}[1] {\textcolor{RedOrange}{{#1}}}

\section{Introduction}\label{sec:intro}

Linearizability~\cite{classic-lin90} is a standard correctness condition for concurrent objects.
It has enjoyed relative success when applied in practice, both since the resulting correctness properties are composable, and since the correctness condition for operations on a linearizable concurrent object can be written as a simplified sequential implementation.
The same cannot be said of objects whose operations exhibit non-linearizable edge cases by design.
Such objects are relatively common in practice, but there is no convenient way to sequentially specify their behavior.

Interval linearizability~\cite{interval-lin18} is the state of the art in relaxing linearizability's atomicity requirement, thus allowing more complex interactions between concurrent operations on the same object, while maintaining linearizability's compositionality guarantees.
Formally it is an expressive and flexible way to reason about non-linearizable behaviors, but its specification style, interval-sequential specifications, is challenging to write.
Interval-sequential specifications work well when multiple non-linearizable operations can be summarized into one aggregate step, but their effectiveness breaks down when an operation takes more than one atomic step.

For instance, consider an operation $wr(x)$ (``Write-Read $x$''), operating on a shared register $r$, which performs the following atomic operations: write $r \gets x$ (write $x$ to $r$); read and respond with $r$'s value.
Concurrent invocations of $wr(x)$ may interleave read/write steps arbitrarily, leading to non-linearizable results such as two concurrent $wr(x)$ and $wr(x')$ both reading $x'$ due to $x$ being overwritten before it can be read back.
An idiomatic interval-sequential specification of $wr(x)$'s behavior is possible, but it relies on a subtle argument that informally requires enumerating which of the read and write permutations are actually possible, and formally requires its own proof.

Extending $wr(x)$ by just one additional atomic read to get $wrr(x)$ (``Write-Read-Read $x$'') is easy to write operationally: write $r \gets x$; read $r_1 \gets r$; read $r_2 \gets v$ and respond with the observed tuple $(r_1, r_2)$.
However, our interval-sequential specification for $wrr(x)$ must be completely reworked to explicitly quantify over read-write interleavings, and becomes more complex.
Given a sequential specification is typically used as an oracle of correctness on which to base in-depth analysis, writing the specification itself should have minimal complexity.
The state of the art does not satisfy this condition.

\begin{figure}
    \centering
    \includegraphics{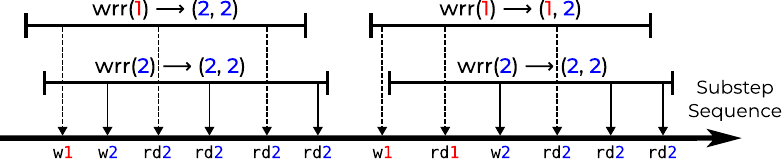}
    \caption{Two valid substep linearizations of concurrent $wrr$ operations. Each operation $wrr(x) \to (y, z)$ is split into 3 substeps which write $x$ ($\mathtt{w}x$), read $y$ ($\mathtt{rd}y$), and read $z$ ($\mathtt{rd}z$) respectively. }
    \label{fig:wrr-intervals}
\end{figure}

To correct this, we introduce \emph{substep linearizability}, a compositionality-preserving relaxation of linearizability's atomicity requirement which we show has power equal to interval linearizability.
Under substep linearizability, each of a concurrent object's operations may be split into an arbitrary non-predetermined number of \emph{substeps}, which may interleave freely with other operations' substeps if the two operations are concurrent.
A substep linearizable object that always uses one substep is a linearizable object.
Crucially, however, a substep-sequential specification may be expressed by taking a sequential specification and dividing it into one or more atomic steps, whose semantics may be specified as directly as when writing a fully sequential specification.
That is, the operations $wr(x)$ and $wrr(x)$ may be specified as we describe them, as two and three atomic steps respectively, while preserving the correctness and compositionality properties we expect of linearizability semantics.
Figure~\ref{fig:wrr-intervals} illustrates two possible substep interleavings when invoking $wrr(1)$ and $wrr(2)$ in parallel.

In this paper, we first introduce X-CCorrectness (``Concurrent Correctness for X''), our formal framework for reasoning about linearizability variants in Section \ref{sec:lin-definitions}.
Section \ref{sec:classic-lin} and beyond define simple linearizability, substep linearizability, and interval linearizability.
We illustrate in Appendix \ref{sec:additional-lins} that X-CCorrectness is suitable for defining other prior work, such as set linearizability~\cite{set-lin94} and weaker properties like sequential consistency and quiescent consistency~\cite{quiescent-consistency-2014,art-of-multiprocessor-programming-2008}.
For substep and interval-sequential linearizability, those subsections contain a detailed comparison of their sequential specification styles using $wr(x)$ and $wrr(x)$.
Section \ref{sec:equiv-compos} defines subsumption and composability properties for X-CCorrectness, and Section \ref{sec:interval-substep} uses those definitions to compare substep and interval linearizability.
Sections \ref{sec:exchanger}--\ref{sec:snapshot-isolation} are case studies with substep-sequential specifications of atomic exchangers~\cite{concurrency-aware-lin15}, a popular non-linearizable queue structure~\cite{cameron314_concurrentqueue}, and strong snapshot isolation for database transactions~\cite{strong-snapshot-2006}.

\section{Definitions of Linearizability}\label{sec:lin-definitions}

We first give a general definition that supports a variety of correctness conditions, including many common notations used in this paper.
We then apply these definitions to specific forms of linearizability: classical (as baseline), substep, and interval.
As part of our precise definitions of substep and interval linearizability, we provide detailed comparisons of their sequential specifications in terms of $wr$ and $wrr$ from Section \ref{sec:intro}.

\subsection{Generalized Concurrent Correctness}\label{sec:x-ccorrectness}

We generalize correctness of concurrent objects and histories as X-CCorrectness.
The condition is shown in Definition~\ref{def:x-lin}, constraining a history $H$ of an object $Obj$, parametrized over total functions $f$ and predicates $P$. If X-CCorrectness satisfies composability (defined in section~\ref{sec:equiv-compos}), then we instead refer to it as X-Linearizability.
Notation and context needed to understand the definition follow.

\begin{definition}[X-CCorrect]\label{def:x-lin}
    A history $H$ is X-CCorrect under $P$ iff $\exists S \in f(H): P(S)$. An object $Obj$ is X-CCorrect if all histories $H\in\H_{Obj}$ are X-CCorrect.
\end{definition}

\textbf{Histories.}
Histories may be either \emph{complete} or \emph{incomplete}.
Incomplete histories contain operations which have an invocation but not a response.
This can occur if a thread crashes in the middle of an operation, or if an ongoing execution is being tested.

A history $H$ is comprised of an ordered sequence of $m$ invocations $i = \langle name, arg, pid\rangle$ and $n \leq m$ responses $r = \langle name, ret, pid\rangle$.
Here, $name$ identifies the operation invoked/responding, $arg$ is an argument, $ret$ is a return value, $pid$ identifies the invoking thread/process.
For convenience, we assume that all invocations and responses have a unique index in $H$.
We can construct this from any $H$ by tagging each $arg$ and $ret$ with that index.
Each response is matched to a unique invocation and $inv(r)$ returns said invocation.
$res(i)$ returns the response matched to $i$ or $\bot$ if none exist.
Let all histories $H \in \H$ satisfy the following criteria, where $<$ totally orders invocations and responses by $idx$:
\begin{itemize}
    \item $\forall j\in 1..m_r:inv(r_j)<r_j$
    \item For all $i = \langle n_i, A, p_i\rangle$ and $r=\langle n_r, R, p_r\rangle$, $inv(r) = i \implies n_i = n_r\land p_i = p_r \land \nexists i' : i = \langle n_i', A', p_i' \rangle \land i < i' < r$
\end{itemize}

In other words, the response to an invocation always occurs after it in the history, they share a $name$ and $pid$, and no operation with the same $pid$ can be invoked between them.

\textbf{Operations.}
For every invocation $i_k$ in a history $H$, the operation $o_k = \langle i_k, res(i_k) \rangle$.
That is, the operation identifies either a matching pair of invocation and response, or if $H$ is incomplete, an invocation that does not have a response where $res(i_k) = \bot$.
We write an operation's invocation as $inv(o_k) = i_k$, and the response as $res(o_k) = res(i_k)$.
Operations are indexed by their invocation, so $idx(o_k) = idx(inv(o_k))$.
For any history $H$ and all $i_k$ in $H$, $ops(H)$ is the set of all operations $\langle i_k, res(i_k) \rangle$.

\textbf{Happens-Before.}
Operations are partially ordered by the happens-before relation $<_H$, such that $o_a <_H o_b$ iff $res(o_a) \neq \bot$ and $res(o_a) < inv(o_b)$ in $H$.

\textbf{Objects.}
Each operation $o_k$ corresponds to some action on some object $obj(o_k)$.
The corresponding invocation and response $o_k = \langle i_k, r_k \rangle$ necessarily refer to the same object, therefore $obj(i_k) = obj(r_k) = obj(o_k)$.
Each object $Obj$ has a set of \emph{well-formed} concurrent histories $\H_{Obj} \subseteq \H$.

\textbf{Restriction.}
When reasoning about compositionality, we often restrict histories or collections of operations to a specific object $A$.
That is, $H|_A$ contains all $i_k, r_k$ in $H$ such that $obj(i_k) = obj(r_k) = A$.
We use the same notation for any collection whose elements are associated with objects in this sense.

\textbf{Parameters.}
X-CCorrectness is parametrized over:
\begin{enumerate}
    \item A total function $f : \H \rightarrow \S$ which maps a history to a set of ordered sequences of events $\{ S_1, S_2 , ..., S_p \}$, where each $S_i = e_1, ..., e_k$.
    The meaning of $e$ is defined by the type of linearizability.
    We assume that each event $e$ has an associated object $obj(e)$ from $H$.
    Similar to histories, we define $S|_A$ as sequences restricted to events associated with object $A$.
    
    \item A predicate $P(S)$ that identifies whether $S$ is a valid sequence of events. 
    For a sequence $S$ that involves two objects $A$ and $B$, we assume that $P$ is defined independently for each object.
    That is, $P(S) = P(S|_A) \land P(S|_B)$.
    We use $S_1\leq_p S_2$ to denote $S_1$ being a prefix of $S_2$.
    We require that $P$ is prefix-preserving in that if $P(S)$ holds then $P(S')$ holds for all $S'\leq_pS$. We place no restriction on $P$, other than it must be a predicate on valid sequences. Let $Pred(X)$ be the set of all valid predicates for X-CCorrectness.
\end{enumerate}

\textbf{History Completions.}
Incomplete histories, where there is an $i$ in $H$ such that $res(i) = \bot$, can be difficult to reason about, so we define their completions.
This definition mirrors those in~\cite{interval-lin18,classic-lin90}.
While history completions are known to disallow some specifications~\cite{compositional-lin23}, they are often useful.

We define the set of completions $C(H)$ of a history $H$ such that:
\begin{itemize}
    \item $C(H)$ contains all invocation-response pairs $\langle i, r \rangle$ in the same relative order as in $H$.
    \item For all invocations $i$ in $H$ where $res(i) = \bot$, either:
    \begin{itemize}
        \item[$\lor$] $i$ is not in $C(H)$.
        \item[$\lor$] $i$ is in $C(H)$ and there is $r \in C(H)$ such that $res(i) = r$ in $C(H)$.
        $r$ may have any value.
    \end{itemize}
    \item $C(H)$ only contains invocations and responses defined in the previous two conditions.
\end{itemize}

To define $P$ over only complete histories, we can define the function $f(H) = \bigcup_{H' \in C(H)} f(H')$.
Unless otherwise stated, we assume complete histories in this paper.

\subsection{Classic Linearizability}\label{sec:classic-lin}

What we call classic linearizability, more commonly referred to as simply \emph{linearizability}, was first introduced by Herlihy and Wing~\cite{classic-lin90}.
In this case, the set of events is exactly the set of operations $ops(H)$.
The function $f(H)$ produces all permutations of $n = |ops(H)|$ operations $o_1, ..., o_n \in ops(H)$ consistent with happens-before in $H$.
That is:
\begin{equation*}
    o_1, ..., o_n \in f(H) \text{ iff } \forall x, y \in 1..n : x < y \implies \lnot (o_x >_H o_y)
\end{equation*}

Then, $P(o_1, ..., o_n)$ is a sequential specification, specifying valid sequences of atomic operations applied to the underlying object.
$P$ is often modeled as a state machine whose transitions are labeled by operations, such that it accepts a subset of sequences $o_1, ..., o_n$.


\subsection{Substep Linearizability}\label{sec:substep-lin}

Substep linearizability is a contribution of this paper.
Like classical linearizability, we consider permutations of the operations implied by a history $H$ that are consistent with happens-before.
Substep linearizability differs by allowing operations to split into multiple atomic parts we call \emph{substeps}.
The last substep taken by an operation may be marked \emph{final}, meaning that extensions of that linearization cannot show the operation taking further substeps.

For any history $H$, we define a substep as $s = \langle o, m \rangle$ where $o \in ops(H)$ and $m \in \{\text{true}, \text{false}\}$ is a final marker.
A substep's operation is written $op(s)$, and its finality is written as the predicate $fin(s)$.
We define $f(H)$ as containing all sequences of substeps $s_1, ..., s_n$ such that:
\begin{align*}
    s_1, ..., s_n \in f(H) \text{ iff } & ((\forall o \in ops(H) : res(o) = \bot \lor \exists x \in 1..n : op(s_x) = o) \\
        & \land (\forall x, y \in 1..n : x < y \implies \lnot (op(s_x) >_H op(s_y))) \\
        & \land \forall x \in 1..n : \\
            & \quad ((fin(s_x) \implies \nexists y \in (x+1)..n : op(s_x) = op(s_y)) \\
            & \quad \land ((\lnot fin(s_x) \land res(op(s_x)) \neq \bot) \\
            & \quad \quad \implies \exists y \in (x + 1)..n : op(s_x) = op(s_y) \land fin(s_y))))
\end{align*}

That is, all completed operations in $ops(H)$ occur at least once in $op(s_1),op(s_2),...,op(s_n)$ (and vice versa), and no operations may disagree with partial history order $<_H$. The sequence $op(s_1),op(s_2),...,op(s_n)$ may contain duplicate operations.
For all substeps involving the same operation, if the operation has a response then its last substep must be marked \emph{final}, meaning extensions of $s_1, ..., s_n$ may not add substeps to that operation.
Incomplete operations do not need a final substep.
The predicate $P(s_1, ..., s_n)$ is called a substep-sequential specification, which defines valid substep interleavings.

Note that per this definition, substep linearizability does not require complete histories.
An incomplete operation may take substeps, meaning that the limitations of history completion do not apply.
However, defining $P$ where response values may be missing is more complex, so we restrict our treatment in this paper to complete histories.
 
\textbf{Annotated Substeps.}
We optionally annotate an operation's substeps with indices to make writing substep-sequential specifications easier.
For any substep sequence $s_1, ..., s_n \in f(H)$, we can count repetitions of each operation $o \in ops(H)$ and assign each one a sequential index $i \in \mathbb{N}$.

Indices give us the option to specify $P$ as $P(e_1, ..., e_n)$, in terms of annotated substeps $e_k$ for $k \in 1..n$, rather than the underlying sequence of substeps.
The annotated substep $e_k$ for some underlying operation $s_k$, given its assigned index $i$, has the form $e_k = \langle op(s_k), fin(i) \rangle$ if $fin(s_k)$, and $e_k = \langle op(s_k), step(i) \rangle$ otherwise.
The advantage of specifying $P$ in this way is that indexing steps is very common in concrete specifications of $P$.
Especially when describing $P$ as a state machine, having access to this information is much easier than extending the state machine to calculate it.
For convenience, we define properties of the underlying operation on $e_k$ and $s_k$ respectively.
For example, assume $res(e_k) = res(s_k) = res(op(s_k))$.

We model $P$ as a possibly non-deterministic state machine with arbitrarily labeled states $q$, where transitions are determined by substeps $e$.
Thus, $P(e_1, ..., e_n)$ becomes the assertion that our state machine accepts the sequence of transitions $e_1, ..., e_n$ starting at an initial state $q_0$.
Acceptable transitions are defined by the relation $\delta(q, e, q')$, which may be written in terms of all information available in $q$ and $e$: any information with which $q$ is labeled, the substep index and finality, and the invocation and response of the underlying operation $o_k$.

See the following sections for simple examples of defining $P$, as well as a proof that substep linearizability is non-blocking in Section~\ref{sec:substep-non-blocking}.
For larger scale examples, see Sections \ref{sec:exchanger}--\ref{sec:snapshot-isolation}.




\subsubsection{Example: Write-Read Substep-Sequentially}

We model $wr$ substep-sequentially as a state machine whose states $q$ are labeled with a current register value $q.r$.
Each operation $o$ is considered as two substeps, $val_i(o)$ is the value passed to the invocation of $o$ ($x$ in the original description), and $val_r(o)$ is the response value.
The transition relation $\delta(q, e, q')$ can be defined as:
\begin{itemize}
    \item If $e = \langle o, step(1) \rangle$, $q'.r = val_i(o)$ \hspace*{\fill} \fbox{\scriptsize Write $r \gets x$}
    \item If $e = \langle o, fin(2) \rangle$, $val_r(o) = q.r \land q'.r = q.r$ \hspace*{\fill} \fbox{\scriptsize Read and respond with $r$'s value}
\end{itemize}

The specification directly matches the original description, as annotated on the right.

\subsubsection{Example: Write-Read-Read Substep-Sequentially}

We extend the substep-sequential specification of $wr$ to $wrr$ by adding one substep to it.
Each state $q$ is still labeled with the current register value $q.r$.
The operation $o$ still has the invocation value $val_i(o)$.
We write operation's two result values as $\mathit{val}_{r1}(o)$ and $\mathit{val}_{r2}(o)$, respectively.
With this, the transition relation $\delta(q, e, q')$ can be defined as:
\begin{itemize}
    \item If $e = \langle o, step(1) \rangle$, $q'.r = val_i(o)$ \hspace*{\fill} \fbox{\scriptsize Write $r \gets x$}
    \item If $e = \langle o, step(2) \rangle$, $\mathit{val}_{r1}(o) = q.r \land q'.r = q.r$ \hspace*{\fill} \fbox{\scriptsize Read $r_1 \gets v$}
    \item If $e = \langle o, fin(3) \rangle$, $\mathit{val}_{r2}(o) = q.r \land q'.r = q.r$ \hspace*{\fill} \fbox{\scriptsize Read $r_2 \gets v$ and respond with $(r_1, r_2)$}
\end{itemize}

Again, this specification can be read as a one to one transcription of our operational description, as annotated on the right.
Producing a behavioral specification like this takes no more effort than writing an imperative program, which gives substep-sequential specifications a similar level of accessibility to fully sequential specifications.

\subsection{Substep Linearizability is Non-Blocking}\label{sec:substep-non-blocking}

Ideally, we would like a definition of linearizability to be \emph{non-blocking} in that no total operation has to wait for a different operation in order to progress.
We say an operation name $N$ is \emph{total} with respect to a predicate $P$ if for all $i=\langle N, arg, pid\rangle$ and $S$ such that $P(S)$ holds, there exists a response $r=\langle N, ret, pid\rangle$ with $o=\langle i, r\rangle$, such that $P(S\cdot \langle o, step(1)\rangle\cdot \langle o, step(2)\rangle\cdot ...\cdot \langle o, fin(k)\rangle)$ holds for some $k\geq 1$.
In other words, there is some response to $i$ such that a finite sequence of substeps of operation $o$ that can be appended to $S$ while maintaining its validity with respect to $P$.

We show this is the case for substep linearizability.
Note that we only consider predicates that hold uniquely for sequences of completed operations, since if the object itself is blocking, there is no point in the definition being non-blocking.
Let $H\cdot r$ be the result of appending a response $r$ to the end of $H$ and $S\cdot e$ be the result of appending a substep to the end of a sequence of substeps.

Theorem~\ref{thm:non-blocking} implies that substep linearizability is non-blocking:

\begin{theorem}\label{thm:non-blocking}
    Let $H$ be a history containing an unmatched invocation $i=\langle N, arg, pid\rangle$ of a total operation name $N$ where $H$ is substep-linearizable with respect to a predicate $P\in Pred(substep)$. Then there exists return value $ret$ such that $H\cdot\langle N, ret, pid \rangle$ where $\langle N, ret, pid \rangle = res(i)$ is in the set $\H$ and is substep-linearizable with respect to $P$.
\end{theorem}

\begin{proof}
    Since $H$ is substep-linearizable, there exists an $S\in f(H)$ such that $P(S)$ holds.
    If $S$ contains a substep of an operation $o$ such that $inv(o_j)=i_j$, then we must have $S\in f(H\cdot \langle N, ret, pid \rangle)$ for some $ret$ with $\langle N, ret, pid \rangle = res(i_j)$ by definition of substep-linearizability, so the Theorem holds.
    If $S$ does not contain a substep of such an operation $o$, then since $N$ is total, there exist $N$, $ret$, $pid$, and a sequence of substeps such that for $o=\langle i, \langle N, ret, pid \rangle\rangle$, $P(S\cdot \langle o, step(1)\rangle\cdot \langle o, step(2)\rangle\cdot ...\cdot \langle o, fin(k)\rangle)$ holds.
    Furthermore, $S\cdot \langle o, step(1)\rangle\cdot \langle o, step(2)\rangle\cdot ...\cdot \langle o, fin(k)\rangle\in f(H\cdot \langle N, ret, pid \rangle)$, so the Theorem holds in this case as well.
\end{proof}

\subsection{Interval Linearizability}\label{sec:interval-lin}

Interval linearizability~\cite{interval-lin18} generalizes linearizability by fusing a set of invocations $I$ and a set of responses $R$ into one step, or \emph{interval} $(I, R)$.
Note that unlike substep linearizability, $f(H)$ is only well defined on complete histories.
The resulting $f(H)$ contains sequences of intervals such that:
\begin{align*}
    I_1, R_1, ..., I_n, R_n \in f(H) \text{ iff } & (\forall \, \langle i, r \rangle \in ops(H) : \\
            & \quad (\exists x, y \in 1..n : i \in I_x \land r \in R_y \\
            & \quad \land \forall x \in 1..n : r \notin I_x \land i \notin R_x) \\
        & \land \forall x, y \in 1..n : \\
            & \quad ( I_x \cap I_y = R_x \cap R_y = \emptyset \\
            & \quad \land \forall i \in I_x, r \in R_y : op(i) = op(r) \implies x \leq y \\
            & \quad \land \forall r \in R_x, i \in I_y : r < i \implies x < y))
\end{align*}

The definition above can be explained via three rules:
\begin{itemize}
    \item Every invocation and response in $H$ must be assigned to exactly one set $I$ and $R$ respectively, where sets $I$ contain only invocations, and sets $R$ contain only responses.
    \item An operation's invocation must occur in an interval no later than its response.
    \item A response that happens-before an invocation must occur in an earlier interval than that invocation.
\end{itemize}

The predicate $P(I_1, R_1, ..., I_n, R_n)$ is called an interval-sequential specification.
It is modeled as a state machine with arbitrarily labeled states $q$, whose transitions are determined by intervals $I, R$.
Acceptable transitions are defined by the relation $\delta(q, I, R, q')$.
More precisely, the sequence $I_1, R_1, ..., I_n, R_n$ is accepted by $P$ if and only if there exists $q_0, q_1, ..., q_n$ such that $q_0$ is a valid initial state and for all $i \in \{1,2,...,n\}$, $(q_{i-1},I_i, R_i,q_i) \in \delta$.

As we comment in Section \ref{sec:intro}, $\delta$ can be challenging to define correctly.
It can work well for operations that combine operations into fewer steps -- see the atomic exchanger case study in Section~\ref{sec:exchanger}.
However, it is challenging to express even simple interval-sequential specifications outside of that situation.
We show this by giving interval-sequential specifications of $wr$ and $wrr$ in the following sections.


\subsubsection{Example: Write-Read Interval-Sequentially}

To give an interval-sequential specification for $wr$ as defined in Section \ref{sec:intro}, we need to map each invocation's two atomic steps onto one interval, interpreting multiple invocations and responses at once (the sets $I$ and $R$ respectively).

First consider two concurrent invocations: $wr(1) \to x$ which writes $1$ and reads $x$, and $wr(2) \to y$ which writes $2$ and reads $y$.
Our interval-sequential specification must specify what values of $x$ and $y$ are possible, so we examine the valid permutations of the operations' 4 atomic steps.
We write the step where $wr(1)$ writes its value as $wr(1)_\text{writes}$, and the step where it reads a value $p$ as $wr(1)_{\text{reads }p}$.

If $x = 1$ and $y = 2$, then the invocations are sequential and not concurrent: this is the effect of $wr(1)_\text{writes}$, $wr(1)_{\text{reads }1}$, $wr(2)_\text{writes}$, $wr(1)_{\text{reads }2}$.
No concurrent execution can have this outcome, since a concurrent execution would start $wr(2)$ earlier, preventing $wr(1)$ from reading its own value.
Our interval-sequential specification does not need to specify this, since two sequential intervals achieve the same result.

In a concurrent execution, sequences where $x = 2$ have the form $wr(1)_\text{writes}$, $wr(2)_\text{writes}$, ..., $wr(1)_{\text{reads } 2}$, ...
This case requires that $y = 2$, since $wr(2)_\text{reads}$ can only occur on either side of $wr(1)_\text{reads}$, nor before $wr(1)_\text{writes}$ since that would not be sequentially consistent.
By swapping $1$ and $2$, the same argument means $y = 1 \implies x = 1$.

Generalizing, writes by concurrent executions of $wr$ can never interleave with reads.
Every concurrent operation must start by writing its own value.
To have $wr(p)_\text{reads}$, we must also have an earlier instance of $wr(p)_\text{writes}$.
For a read to be followed by a write, that write must belong to a new non-concurrent invocation of $wr$.
That is, for an operation to not read its own write, the only concurrent interleaving looks like $wr(v_1)_\text{writes}$, $wr(v_n)_\text{writes}$, ..., $wr(v_1)_{\text{reads } v_n}$, ..., $wr(v_m)_{\text{reads } v_n}$, where $n$ and $m$ are the sizes of $I$ and $R$ respectively.

Given this insight, our interval-sequential specification of $wr$ is an automaton with one state $q$.
Each invocation $i$ has an argument $\mathit{val}(i)$ ($x$ in the original description), and every response $r \in R$ has a response value $\mathit{val}(r)$.
That is, the invocation looks like $wr(\mathit{val}(i)) \to \mathit{val}(r)$.
In this case, the transition relation $\delta(q, I, R, q)$ can be defined iff there exists $i \in I$ such that for all $r \in R$, $val(r) = val(i)$.

This specification treats invocations as write operations and responses are read operations.
For any interval $(q, I, R, q)$, every response in $R$ must read the value written by some invocation in the preceding set of invocations, and every response in $R$ must read the same value.
It is not satisfying, however.
As we show above, it takes nontrivial reasoning to claim this specification truly represents the intended behavior of $wr$.
Observing this property during analysis given a clear statement of intended behavior makes sense.
Requiring the derivation of such a property in order to write an initial specification does not.

\subsubsection{Example: Write-Read-Read Interval-Sequentially}

Unlike $wr$, to divide $wrr$ into sets of invocations and responses we must explicitly model the possible interleavings of the three component operations.

Consider a pair of concurrent operations $wrr(1)$ and $wrr(2)$.
Aside from completing each operation sequentially, there are four valid outcomes.
Figure \ref{fig:wrr-intervals} illustrates the first two.
If $wrr(1)$ writes first, we get these interleavings: $wrr(1)_\text{writes}$, $wrr(2)_\text{writes}$, $wrr(*)_{\text{reads }2}$ ($4\times$) -- meaning $wrr(1)$ and $wrr(2)$ both return $(2, 2)$; and $wrr(1)_\texttt{writes}$, $wrr(1)_{\text{reads }1}$, $wrr(2)_\text{writes}$, $wrr(*)_\text{reads 2}$ ($3\times$) -- meaning $wrr(1)$ returns $(1, 2)$ and $wrr(2)$ returns $(2, 2)$.
If $wrr(2)$ writes first, we get the same pattern but with $1$ and $2$ swapped.
The result $(1, 2)$ being possible proves our property for $wr$ does not apply to $wrr$.
Instead, we must directly express possible sequences of interleavings.

More formally, an interval-sequential specification of $wrr$ is a non-deterministic state machine with states $q$ labeled with the set of pending operations $q.pend$ and $q.obs$ a function from any operation in $q.pend$ to either a value or $\bot$.
Invocations $i \in I$ have values $\mathit{val}(i)$, and responses $r \in R$ have two values each, $\mathit{val}_1(r)$ and $\mathit{val}_2(r)$.
For $I$ and $R$, $ops(I)$ and $ops(R)$ are the sets of corresponding operations respectively.
The transition relation $\delta$ can be defined such that $(q, I, R, q') \in \delta$ iff there exists a permutation $s$ of all elements of $I$, such that:
\begin{itemize}
    \item Let $pos(i)$ be the position of $i \in I$ in the sequence $s$.
    \item Let $obs(x)$ be a function from $x \in q.pend \cup ops(I)$ to a value or $\bot$ such that:
    \begin{itemize}
        \item Where $q.obs(x) \neq \bot$, $obs(x) = q.obs(x)$
        \item Where $q.obs(x) = \bot$, non-deterministically either $obs(x) = \bot$, or $obs(x) = val(i)$ where $i \in I$ such that if $inv(x) \in I$ then $pos(inv(x)) \leq pos(i)$.
    \end{itemize}
    \item For all $r \in R$:
        \begin{itemize}
            \item[$\land$] $obs(op(r)) \neq \bot$ and $\mathit{val}_1(r) = obs(op(r))$
            \item[$\land$] $\mathit{val}_2(r) = val(i)$ where $i$ is the last invocation in $s$
        \end{itemize}
    \item $q'.pend = q.pend \cup ops(I) \setminus ops(R)$
    \item $q'.obs(x) = obs(x)$, where $x \in q'.pend$
\end{itemize}

Intuitively, the resulting state machine non-deterministically tracks the middle step of each operation using $obs(x)$.
During any substep in which an operation is pending, some ``middle'' value may optionally be read into $obs(x)$, meaning the middle step occurs at any point between invocation and response of its corresponding operation.
We implicitly discard cases where $obs(x)$ stays $\bot$, since the response cannot be $\bot$.
Following the selected sequence $s$, $obs(x)$ observes only invocations on or after $inv(x)$.

While this specification works, it has significant incidental complexity.
For example, a previous draft allowed $obs(x)$ to capture invocations before $inv(x)$, which was missed for some time by multiple co-authors.
Furthermore, the specification starts to resemble manually keeping track of all possible substeps, which is the use case for substep linearizability.

\section{Defining Subsumption and Composability}\label{sec:equiv-compos}

In this section, we establish two definitions important for creating definitions of linearizability: Subsumption and composability.
Intuitively, a definition of CCorrectness subsumes another if it can define an equivalent predicate for all predicates of the other definition.
To fit with the common meaning of linearizability, we want our definitions to be composable.
That is, for any history composed of multiple objects, if the histories restricted to each individual object are X-CCorrect under a predicate, then the entire history is X-CCorrect under the same predicates and vice-versa.

First, Theorem~\ref{thm:substep-composable} states that substep linearizability is composable.
Theorem~\ref{lem:compose-conditions} generalizes conditions for composability of X-CCorrectness, and Theorem~\ref{thm:equiv-compose} states that composability follows from being subsumed by a composable definition.

For two definitions of CCorrectness, X-CCorrectness and Y-CCorrectness, we say that Y-CCorrectness \emph{subsumes} X-CCorrectness if for all $P_x\in Pred(X)$, there exists a $P_y\in Pred(Y)$ such that $H$ is X-CCorrect with respect to $P_x$ if and only if $H$ is Y-CCorrect with respect to $P_y$.
We say that X-CCorrectness and Y-CCorrectness are \emph{equivalent} if each subsumes the other.

\begin{definition}[Subsumes]
    Y-CCorrectness \emph{subsumes} X-CCorrectness if $\forall P_X\in Pred(X):\exists P_Y\in Pred(Y):\forall H\in\H:(\exists S\in f_X(H):P_X(S))\iff (\exists S'\in f_Y(H):P_Y(s))$.
\end{definition}

\begin{definition}[Equivalence]
    X-CCorrectness is \emph{equivalent} to Y-CCorrectness if X-CCorrectness subsumes Y-CCorrectness and Y-CCorrectness subsumes X-CCorrectness.
\end{definition}

As a reminder, we require $P(S)=P(S|_A)\,\land\,P(S|_B)$ for a sequence composed of operations on 2 objects, $A$ and $B$.
This can easily be extended to histories composed of more than 2 objects, but we only handle this case for simplicity.

\begin{definition}[Composability]\label{def:composability}
    A history $H$ is \emph{composable} under X-CCorrectness if $\forall P\in Pred(X):(((\exists S_A\in f(H|_A):P(S_A))\,\land\, (\exists S_B\in f(H|_B):P(S_B)))\iff \exists S\in f(H):P(S))$.
    X-CCorrectness is composable if all possible histories $H\in \H$ are composable.
\end{definition}

\begin{theorem}\label{thm:substep-composable}
    Substep linearizability is composable.
\end{theorem}
This theorem can be proven using the same proof method as for classical linearizability. Intuitively, if there exist an $S\in f(H)$ that satisfies $P$, then $S|_A$ and $S|_B$ both satisfy $P$ and are in $H|_A$ and $H|_B$, respectively since these sub-histories maintain the same happens-before ordering of $H$. Furthermore, given an $S_A\in f(H|_A):P(S_A)$ and $S_B\in f(H|_B):P(S_B)$, an $S\in f(H):P(S)$ can be constructed by interleaving $S_A$ and $S_B$ in such a way that does not violate happens-before. It can be shown that such an interleaving must exist by proving that the transitive closure between the happens-before ordering of operations in $H$ with the ordering of operations in $S_A$ and $S_B$ cannot form a cycle.

We show that a history $H$ being composable can be reduced to saying that:
(1) For each object $A$ and for all sequences $S$ generated from $f(H)$, there is some sequence $S_A$ generated from $f(H|_A)$ such that $S_A$ is a prefix of $S|_A$ 
(2) for each set of sequences generated from each object, there is some sequence $S$ generated from $f(H)$ such that $S|_A$ is a prefix of the sequence $S_A$ generated from each object $A$.
The proof follows from the from the definition of a predicate on a sequence composed on multiple objects.

\begin{theorem}\label{lem:compose-conditions}
    A history $H$ is composable under X-CCorrectness if and only if the following conditions hold:
    \begin{enumerate}
        \item $\forall S\in f(H):(\exists S_A\in f(H|_A): S_A\leq_pS|_A)\land(\exists S_B\in f(H|_B): S_B\leq_pS|_B)$
        \item $\forall S_A\in f(H|_A), S_B\in f(H|_B): \exists S\in f(H): S|_A\leq_pS_A\land S|_B\leq_pS_B$
    \end{enumerate}
\end{theorem}

\begin{proof}
    We first show conditions 1 and 2 imply $H$ is composable under X-CCorrectness.
    Given $P\in Pred(X)$ without loss of generality, we show the bidirectional implication in Definition~\ref{def:composability} holds for X.
    Assuming the left-hand side of Definition~\ref{def:composability} holds, then by condition 2 there exists an $S\in f(H)$ such that $S|_A\leq_pS_A\land S|_B\leq_pS_B$ holds.
    Therefore, since $P$ is prefix-preserving and $P(S_A)\land P(S_B)$ holds, we have that $P(S|_A)\,\land\,P(S|_B)$ holds which by definition is $P(S)$ so the right-hand side must be true.
    Assuming the right-hand side, then by condition 1 there exist $S_A$, $S_B$ such that $S_A\leq_pS|_A\land S_B\leq_pS|_B$.
    Thus, since $P(S)=P(S|_A)\land P(S|_B)$, the left-hand side must be true.

    Secondly, we show that the negation of either condition 1 or 2 implies $H$ is not composable under X-CCorrectness.
    \textbf{(a)} Negating condition 1, there must exist an $S \in f(H)$ such that $\forall S_A\in f(H|_A) : S_A \not\leq_p S|_A$ or $\forall S_B \in f(H|_B) : S_B \not\leq_p S|_B$.
    We choose $P \in Pred(X)$ where $P(S') = S'|_A \leq_p S|_A \land S'|_B \leq_p S|_B$.
    $P(S)$ holds by definition, but the definition of $P$ combined with the negation of condition 1 imply that $\forall S_A\in f(H|_A):\lnot P(S_A)$ or $\forall S_B\in f(H|_B):\lnot P(S_B)$. This contradicts definition~\ref{def:composability}.
    \textbf{(b)} Negating condition 2, there exist $S_A \in f(H|_A)$ and $S_B \in f(H|_B)$ such that $\forall S \in f(H) : S|_A \not\leq_p S_A \lor S|_B \not\leq_p S_B$.
    We choose $P \in Pred(X)$ where $P(S') = S'|_A \leq_p S_A \land S'|_B \leq_p S_B$.
    By definition of $P$, $P(S_A) \land P(S_B)$ must hold, but the negation of condition 2 implies that $\forall S\in f(H): \lnot P(S|_A)\lor \lnot P(S|_B)$.
    Therefore, $\forall S\in f(H):\lnot P(S)$ contradicts Definition~\ref{def:composability}.
    
\end{proof}

Theorem~\ref{thm:equiv-compose} shows that a definition of CCorrectness is composable if it is subsumed by a definition known to be composable.
Intuitively, for any predicate in the original definition, there must be an equivalent predicate in the composable definition.

\begin{theorem}\label{thm:equiv-compose}
    If Y-CCorrectness subsumes X-CCorrectness and Y-CCorrectness is composable, then X-CCorrectness is composable.
\end{theorem}


\begin{proof}
    For a proof by contradiction, assume Y-CCorrectness is composable and subsumes X-CCorrectness, and X-CCorrectness is not composable.
    Since X-CCorrectness is not composable, there exist some history $H$, predicate $P\in Pred(X)$ and objects $A,B$ such that either \textbf{(a)} $\exists S_A\in f_X(H|_A):P(S_A)\,\land\, \exists S_B\in f_X(H|_B):P(S_B)$ but $\forall S\in f_X(H):\lnot P(S)$ or \textbf{(b)} $\exists S\in f_X(H):P(S)$ but $\forall S_A\in f_X(H|_A):\lnot P(S_A)\,\lor\, \forall S_B\in f_X(H|_B):\lnot P(S_B)$.
    By subsumption of X-CCorrectness, there must exist a $P'\in Pred(Y)$ such that $\forall H_0\in\H:(\exists S\in f_X(H_0):P(S))\iff (\exists S'\in f_Y(H):P(S'))$. 
    Then in case \textbf{(a)}, we have that $\exists S_A'\in f_Y(H|_A):P'(S_A')\,\land\, \exists S_B'\in f_Y(H|_B):P'(S_B')$ and $\forall S'\in f_Y(H):\lnot P'(S')$.
    In case \textbf{(b)}, we have that $\exists S'\in f_Y(H):P'(S')$ and $\forall S_A'\in f_Y(H|_A):\lnot P'(S_A')\,\lor\, \forall S_B'\in f_Y(H|_B):\lnot P'(S_B)$.
    Either case contradicts the definition of composability for Y-CCorrectness.
\end{proof}

\section{Relationship Between Substep and Interval Linearizability}\label{sec:interval-substep}

In this section, we first prove that substep linearizability subsumes interval linearizability.
This shows that any interval-sequential specification can be expressed using substeps.
We also show that interval linearizability subsumes substep linearizability under a restricted class of predicates that we describe in Definition~\ref{def:nb-preds}.

We begin by proving that substep linearizability subsumes interval linearizability.
In this direction, we can treat the invocations and responses in an interval linearization as pairs of substeps.
This means that any predicate defined for interval linearizability can be translated directly to a predicate defined for substep linearizability. Let $f_I(H)$ and $f_S(H)$ be the function $f(H)$ defined for interval and substep linearizability, respectively. 

\begin{theorem}\label{lem:sub-subsumes-int}
    Substep linearizability subsumes interval linearizability.
\end{theorem}

\begin{proof}
    Let $H\in\H$ be some history and $P_I\in Pred(interval)$ be some arbitrary predicate defined without loss of generality.
    We prove the existence of a predicate $P_S\in Pred(substep)$ such that $H$ is substep-linearizable with respect to $P_S$ iff $H$ is interval-linearizable with respect to $P_I$.
    For an arbitrary sequence of substeps $S\in f_S(H)$ such that each substep occurs exactly twice, we define a sequence of alternating sets of invocations and responses $S_0$.
    First, combine maximal adjacent sets of the first occurrence of $o_K$'s in $S$ into sets $I_1,I_2,...,I_n$.
    Combine maximal adjacent sets of the second occurrence of $o_K$'s in $S$ into sets $R_1,R_2,...,R_n$.
    Let $S_0=I_1,R_1,I_2,R_2,...,I_n,R_n$.
    Let $P_S(S)$ be the predicate that holds if and only if the following conditions hold:
    \begin{itemize}
        \item Each $o_k$ in the history $H$ appear exactly twice in $S$
        \item $P_I(S_0)$
    \end{itemize}
    The equivalence between $P_I$ and $P_S$ follows from the definitions of interval and substep linearizability.
\end{proof}


Note that the function $f(H)$ for substep linearizability allows for substeps of incomplete operations, while $f(H)$ for interval linearizability only allows sequences such that each invocation has a matching response.
Therefore, in general we can define predicates that are substep-linearizable but not interval-linearizable.
However, we define a class of predicates $NBPred(substep)\subset Pred(substep)$ such that interval-linearizability subsumes substep linearizability over these predicates.
Intuitively, this set defines predicates $P$ such that if $P$ is substep-linearizable with respect to an incomplete history $H$, it is also substep linearizable with respect to a history $H'\in C(H)$.

\begin{definition}[Non-blocking predicates]\label{def:nb-preds}
    A Non-blocking predicate $P\in NBPred(substep)$ is a predicate
    \begin{equation*}
    P\in Pred(substep):\forall H\in\H: (\exists S\in f(H):P(S)\implies
    \exists H'\in C(H):\exists S'\in f(H'):P(S'))
    \end{equation*}
    
\end{definition}

Even with this restriction, finding a predicate defined for interval linearizability that is equivalent to an arbitrary predicate defined for substep linearizability is less straightforward since a specification could require more than two substeps. However, a key observation is that some sequence $S\in f_I(H)$ can be interpreted as representing each of the completions of the original history $C(H)$. In other words, no information about the original history is lost when computing $f(H)$. This allows us to test any substep predicate $P_S$ using the set of sequences $f_I(H)$.

\begin{theorem}\label{lem:int-subsumes-step}
    Interval linearizability subsumes substep linearizability under non-blocking predicates.
\end{theorem}

\begin{proof}
    Let $H\in\H$ be some history and $P_S\in NBPred(substep)$ be some arbitrary predicate defined without loss of generality.
    We prove the existence of a predicate $P_I\in Pred(interval)$ such that $H$ is substep-linearizable with respect to $P_S$ if and only if $H$ is interval-linearizable with respect to $P_I$.
    For an arbitrary sequence $S_0=I_1,R_1,I_2,R_2,...,I_m,R_m\in f_I(H)$, we define the history $H_{S_0}$ as some ordering of the invocations and responses in $S_0$ while maintaining that if $a$ happens before $b$ in $H_{S_0}$, then $b$ was in the same set as $a$ or a later one in the sequence $S_0$.
    Since all possible choices of $H_{S_0}$ result in the same happens-before order between operations, the result is deterministic.
    Let $P_I$ be the predicate that holds for a sequence $S'$ if and only if $H_{S'}$ is substep linearizable with respect to $P_S$.
    We prove that this predicate holds for $H$ if and only if $P_S$ holds for $H$.
    If $\exists S\in f_S(H):P_S(S)$, then by definition of $NBPred$, for some $H'\in C(H)$, $\exists S'\in f(H'):P_S(S')$.
    By definition of interval linearizability, we have that $\forall H_C\in C(H): \exists S_0\in f_I(H) :H_{S_0}=H_C$.
    Let $S_0$ be the sequence such that $H_{S_0}=H'$.
    By definition, if $\exists S'\in f_S(H'):P_S(S)$ then $\exists S\in f_S(H_{S_0}):P_S(S)$ which is the definition of $P_I(S_0)$. Since $\exists S'\in f_I(H):S'=S_0$, we have that $\exists S'\in f_I(H): P_I(S')$.
    Finally, we assume $\exists S'\in f_I(H): P_I(S')$ and prove $\exists S\in f_S(H):P_S(S)$.
    By definition of $P_I$, $\exists S\in f_S(H_{S'}): P_S(S)$.
    But since by construction, $H_{S'}$ preserves the happens-before order of operations in $H$, $S\in f_S(H)$.
    Therefore, $\exists S\in f_S(H):P_S(S)$ which completes the proof.
\end{proof}




\section{Case Study: Atomic Exchanger}\label{sec:exchanger}

The atomic exchanger is a common example of a non-linearizable object, which can be specified as set-linearizable and motivates the concurrency-aware linearizability proof system~\cite{concurrency-aware-lin15}.
It offers one operation $exchange(v)$ which either atomically exchanges $v$ with a concurrent operation, producing pairs of invocations such that $exchange(v_1) = v_2 || exchange(v_2) = v_1$, or $exchange(v_1) = \bot$ and $exchange(v_2) = \bot$.
That is, either two operations successfully exchange values atomically, or no exchange can occur and failure is indicated via $\bot$. 

This case study has already been explored by the state of the art, and is known to be set- and interval-sequentially specifiable.
We first show an interval-sequential specification, then show that a similarly convenient substep-sequential specification is possible.

Our interval-sequential specification is a state machine with unlabeled states $q$.
Each invocation $i$ has a proposed value $prop(i)$, and each response $r$ has a received value $recv(r)$.
Transitions are accepted by $\delta(q, I, R, q')$ such that either (1) $I = \{ i \}$, $R = \{ res(i) \}$, and $recv(res(i)) = \bot$; or (2) $I = \{ i_1, i_2 \}$, $i_1 \neq i_2$, $R = \{ res(i_1), res(i_2) \}$, $prop(i_1) = recv(res(i_2))$, and $prop(i_2) = recv(res(i_1))$.
That is, any pair of concurrent operations may exchange values, or a non-paired operation may fail and receive $\bot$.
We only need to specify pairs of invocations and responses, since larger numbers of concurrent operations are automatically decomposed sequences of intervals by interval semantics.

A substep-sequential specification can emulate an interval-sequential one, as shown in Theorem~\ref{lem:sub-subsumes-int}.
When the intervals would allow groups of operations to execute as one (pairs in this case), we propose a convenience notation for substep-sequential state machines.
Rather than specify each substep individually, we accept atomic sequences of substeps.
That is, we define $\delta(q, S, q')$ where $S$ is a non-empty sequence of substeps, such that if $S = e_1, ..., e_n$, then the state machine accepts the sequence of transitions $q, e_1, ..., e_n, q'$ iff $\delta(q, S, q')$.
If the intention is to accept $e_1, ..., e_n$ and none of its permutations, then this notation is simpler than specifying each step individually and logically ruling out other invalid sequences.

With this notation, we can specify an atomic exchanger substep-sequentially as a state machine with unlabeled states $q$.
Each operation $o$ has a proposed value $prop(o)$ and a received value $recv(o)$.
We define $\delta(q, S, q')$ such that:
\begin{itemize}
    \item If $S = \langle o_1, step(1) \rangle, \langle o_2, fin(1) \rangle, \langle o_1, fin(2) \rangle$ and $o_1 \neq o_2$, then $prop(o_1) = recv(o_2)$ and $prop(o_2) = recv(o_1)$.
    \item If $S = \langle o, fin(1) \rangle$, then $recv(o) = \bot$.
\end{itemize}

The success case is an atomic triple of substeps: $o_1$ takes two steps that surround $o_2$, meaning that they must be concurrent.
Such concurrent operations may exchange values if the proposed and received values match.
Swapping $o_1$ and $o_2$ produces an equivalent symmetrical behavior.
The failure case is a single atomic step which may occur at any time.
As with the interval-sequential specification, cases with more than two concurrent operations are implicitly decomposed into sequential instances of the specified patterns.

The specification above is easy to translate into single steps.
Multi-step sequences can be encoded using state labels storing position within the step sequence and information needed by multi-step assertions.
For the exchanger, an implementation of mutual exclusion with explicit ``propose'' and ``accept'' labels is sufficient.

\section{Case Study: Cameron Desrochers' Queue}

High performance queue data structures are often necessary in multithreaded applications.
A common use case is managing thread pools.
Cameron's queue~\cite{cameron314_concurrentqueue} is a popular\footnote{
    It has 12.5k stars on GitHub as of writing and is a prominent search result for ``lock free queue C++''.
} C++ library whose design aims for competitive performance in scenarios where fully linearizable behavior is not necessary.
The queue respects insertion order per producer (thread that enqueues), but not globally since many use cases do not require it.

Since it relies heavily on operations with multiple visible side-effects, an interval-sequential specification would be difficult to write for similar reasons to our $wrr$ example in Section \ref{sec:interval-lin}.
We show the applicability of substep-sequential specification by specifying three representative operations: enqueue by producer, dequeue, and approximate size.

Per their documentation, Cameron's queue may be considered as a list of sub-queues indexed by producer.
Each producer enqueues to its own sub-queue, and dequeue operations non-deterministically select the head of any sub-queue, or $\bot$ if all sub-queues non-atomically appear empty.
While enqueue is a single atomic operation, dequeue is not.
It performs a non-atomic check on each sub-queue, and if it is possible to witness all sub-queues transiently be empty (even if they were never simultaneously empty), then $\bot$ may be produced.
Similarly, approximate size is specified to perform a non-atomic scan of sub-queue sizes, and may return any sum of sub-queue sizes that individually occur during the operation.

For a set of producers $P$, we specify Cameron's queue as a substep-sequential state machine.
For each operation $o$, we define its name $name(o) \in \{enqueue, dequeue, size\}$.
All operations have a value $val(o)$, which is either argument of an enqueue, the return value of a dequeue, or the return value of a size operation.
$prod(o) \in P$ is an enqueue operation's producer.
States $q$ are labeled with a total function from producers $p \in P$ to sub-queues $q.s(p) = v_1, ..., v_n$, where a sub-queue is a sequence of arbitrary values $v_1, ..., v_n$.
For every operation $o$ and producer $p$, $q.z(o, p) = n$ where $n \in \mathbb{N} \cup \{ \bot \}$.
That is, $q.z$ tracks a per-operation mapping from producer to witnessed sub-queue size, which is $\bot$ when no size has been witnessed.
In the initial state $q_0$, $q_0.s(p) = \emptyset$ for all $p$, and $q_0.z(o, p) = \bot$ for all $o$ and $p$.
We define $\delta(q, e, q')$ such that:
\begin{itemize}
    \item If $e = \langle o, fin(1) \rangle$, $name(o) = enqueue$, and $q.s(prod(o)) = v_1, ..., v_n$ then:
    \begin{itemize}
        \item[$\land$] $q'.s = q.s$ except $q'.s(prod(o)) = v_1, ..., v_n, val(o)$
        \item[$\land$] $q'.z = q.z$
    \end{itemize}
    \item If $e = \langle o, fin(1) \rangle$, $name(o) = dequeue$, and $(\exists p \in P : q.s(p) = v_1, ..., v_n) \land (v_1 = val(o))$ then:
    \begin{itemize}
        \item[$\land$] $q'.s = q.s$ except $q'.s(p) = v_2, ..., v_n$
        \item[$\land$] $q'.z = q.z$
    \end{itemize}
    \item If $e = \langle o, step(*) \rangle$, $name(o) \in \{dequeue, size\}$, there exists $p \in P$ such that $q.z(o, p) = \bot$ and $q.s(p) = v_1, ..., v_n$, then:
    \begin{itemize}
        \item[$\land$] $q'.z = q.z$ except $q'.z(o, p) = n$
        \item[$\land$] $q'.s = q.s$
    \end{itemize}
    \item If $e = \langle o, fin(*) \rangle$, $name(o) = dequeue$, $val(o) = \bot$, and $\forall p \in p : q.z(o, p) = 0$, then $q' = q$.
    \item If $e = \langle o, fin(*) \rangle$, $name(o) = size$, $\forall p \in P : q.z(o, p) \neq \bot$, and $val(o) = \sum_{p \in P} q.z(o, p)$, then $q' = q$.
\end{itemize}

The enqueue operation is atomic and takes a single substep to add the element $val(o)$ to the producer's sub-queue.
The success case of the dequeue operation takes a single substep when one of the sub-queues is observed to have $val(o)$ at its head, in which case that element is dequeued.
Otherwise, dequeue and size share $|P|$ substeps: they non-deterministically observe each sub-queue's size one by one and accumulate the results into $q.z$.
For dequeue, the operation may report $\bot$, indicating a possibly empty queue, if all sub-queues transiently have size $0$.
For the size operation, the observed sizes must sum to the operation's reported value.

As an extension, consider that Cameron's queue also has optimized bulk enqueue and dequeue operations which operate on sequences of elements to enqueue or dequeue at once.
These operations are non-atomic, and can also be given substep-sequential specifications.
The bulk enqueue repeats the single-value semantics for each value it enqueues, assigning each value a substep.
The bulk dequeue operation is involves a request for $M$ elements, as well as a sequence of $N \leq M$ actually dequeued elements.
Comparably to the $\bot$ rule for single dequeues, if $N < M$, it must be possible to non-atomically view a sequence of sub-queues whose length sums to $N$, and which contain the dequeued elements in the right order.

An advantage of substep-sequential specifications in practice is that specifying these operations should not interfere much with the specification above.
They may involve additional state labels or operation types, but in principle the change should be additive.

\section{Case Study: Database Transactions with Strong Snapshot Isolation}\label{sec:snapshot-isolation}

The semantics of concurrent database transactions are often studied~\cite{ansi-sql-1995,jepsen-snapshot-isolation} when describing the allowable behaviors of a given database implementation under concurrency.
We show a substep-sequential specification of strong snapshot isolation~\cite{strong-snapshot-2006}, where each operation is a transaction, and snapshots be placed at any point between the transaction's invocation and response.
Note that assuming strong snapshot isolation is a convenience.
Schultz and Demirbas~\cite{mongodb-transactions-2025} specify a more complete transaction system as totally ordered sequences of steps, which may also be viewed as a substep-sequential specification.

A transaction is a sequence of reads and writes of values $v \in V$ in relation to a set of keys $K$.
We call reads and writes \emph{effects}.
A read has the form $read(k) \to v$, and a write has the form $write(k) := v$.
Transactions are assumed to commit once all reads and writes have completed.
We ignore aborted transactions since snapshot isolation never allows reading aborted writes.
These semantics are broadly similar to existing studies of database semantics, and make no particular assumptions about architecture.

Strong snapshot isolation can be informally defined as the situation where otherwise linearizable transactions exhibit \emph{write skew}, and concurrent overlapping writes are not permitted (in practice they abort).
Consider two concurrent transactions $t_1$ and $t_2$ that read overlapping keys $x$ and $y$ (which we assume refer to integers here).
Assume $t_1$ writes $x - y$ to $x$ and then commits.
Then, $t_2$ writes $x - y$ to $y$ and also commits.
The result is non-linearizable, since neither instance of $x - y$ accounted for the other transaction's update of $x$ or $y$ respectively.
Linearizably, $t_1$ should have written $x - (x - y)$, or $t_2$ should have written $(x - y) - y$.
If the two transactions both wrote to $x$ or $y$, the execution would be invalid.

Our substep-sequential specification of this phenomenon is a state machine with states $q$, labeled by a function $q.kv : K \to (V \times \mathbb{Z})$ mapping each key to a value and a version number, and a function $q.s$ from operations to either a snapshot in the partial function $K \rightharpoonup (V \times \mathbb{Z})$ or $\bot$ indicating no snapshot.
For convenience value-version pairs may be decomposed as, $value(\langle v, * \rangle) = v$ and $version(\langle *, z \rangle) = z$.
Initially, $q_0.kv$ assigns an arbitrary value and version number 0 to every key, and $q.s$ assigns $\bot$ to every operation.
Each operation has an associated transaction $t(o)$ which is a sequence of effects.
We write the set of all keys read by reads in $t(o)$ as $reads(o)$.
Note that this is not just the set of $k$ where $read(k) \to *$ is in $t(o)$: $reads(o)$ excludes reads where $k$ is overwritten before being read.
The transition relation $\theta(q, e, q')$ is defined such that:
\begin{itemize}
    \item If $e = \langle o, step(1) \rangle$, then $q'.s(o) = s$, where $s$ is defined for all $k \in reads(o)$, $s(k) = q.kv(k)$, and undefined otherwise.

    \item If $e = \langle o, fin(2) \rangle$, then let $f_1, ..., f_n$ be the sequence of effects in $t(o)$, there exists $kv_0, ..., kv_n \in K \rightharpoonup V$, inductively defining $kv_{p+1}$ in terms of $f_{p+1}$ and $kv_p$, with initially $kv_0 = \emptyset$, such that:
    \begin{itemize}
        \item If $f_{p+1} = read(k) \to v$, then $v = kv_p(k)$ if $kv_p(k)$ is defined, otherwise $v = value(q.s(o)(k))$.
        In either case, $kv_{p+1} = kv_p$.
        \item If $f_{p+1} = write(k) := v$, then $kv_{p+1} = kv_p$ except $kv_{p+1}(k) = v$.
    \end{itemize}
    Given the final $kv_n$, for all $k \in K$, $q'.kv(k) = (kv_n(k), version(q.kv(k))+1)$ and $version(q.kv(k)) = version(q.s(o)(k))$ where $kv_n(k)$ is defined, and $q'.kv(k) = q.kv(k)$ otherwise.
\end{itemize}

Intuitively, our specification splits each transaction into two steps.
First, we read a snapshot $q.s(o)$ of the database's relevant keys $read(o)$ in one step.
This substep marks where the snapshot occurs, on which the later atomic write will base its view of the database, including which version of each key was observed.
The second substep is specified as an interpreter for effect sequences, accumulating the set of writes performed as $kv_p$.
A transaction can read its own writes according its internal order of operations, so we specify read operations both in terms of the snapshot $q.s(o)$ and already-performed writes $kv_p$.
The final $kv_n$ contains each key-value pair that is atomically written to $q'.kv$, iff the value being overwritten in $q.kv$ has the same version as the value in $q.s(o)$.
If the version differs between snapshot and write, there were concurrent writes to the same key and the execution must be rejected.

\section{Related Work}

Beyond the definitions of linearizability discussed in this paper~\cite{interval-lin18, classic-lin90, set-lin94}, many variants have been proposed in the literature.
Derrick et al.~\cite{quiescent-consistency-2014} define and give a proof method for quiescent consistency, originally proposed by Shavit \& Herlihy~\cite{art-of-multiprocessor-programming-2008}, which allows operations to execute atomically in any order as long as at least one operation is pending.
We show quiescent consistency is expressible using X-CCorrectness in Appendix~\ref{sec:quiescent-consistency}.
Haas et al.~\cite{local-lin16} define local linearizability for a restricted category of containers, which asserts linearizability when the history is projected onto values coming from any specific thread.
Rinberg \& Keidar~\cite{intermediate-lin23} define intermediate-value linearizability in which numerical values can vary non-linearizably within a range.
Hemed et al.~\cite{concurrency-aware-lin15} formalize set-linearizability with their definition of concurrency-aware linearizability and show a technique for verifying said definition.
All of these have restrictive conditions, are not composable, or are less expressive than interval linearizability.

Goubault et al.~\cite{goubault18} argue that any reasonable definition of linearizability satisfies a set of axioms.
Using these axioms, they show that interval-linearizability is the weakest possible correctness condition for concurrent histories, which by our proof of equivalence would imply the same for substep linearizability.

Oliveira Vale et al.~\cite{compositional-lin23} show that linearizability and modular composition are intrinsically related, and Eashan et al.~\cite{complete-prog-lin-2026} apply the concept to program logics.
They define compositional linearizability, a relaxation of linearizability equivalent to interval linearizability with incomplete histories.
Their work improves modular verification, whereas substep linearizability improves sequential specifications.
We leave integration between the two and investigation of substep linearizability without history completion as future work.

\section{Conclusion}

In this paper we present substep linearizability, a correctness condition that allows for operations to take effect at multiple atomic points.
We show how this definition is composable, non-blocking and equivalent to interval linearizability in its ability to capture concurrent behavior.
We also give examples of how the specification of many objects are simpler to represent using substep linearizability.
We leave as future work the development and evaluation of proof systems and model checking techniques to verify conformance with substep-sequential specifications in practice.
However, we are optimistic that substeps having the same basis in sequences of atomic steps as plain linearizability will enable significant re-use of existing techniques.

\bibliography{main}

\appendix

\section{Additional X-CCorrectness Forms}\label{sec:additional-lins}

\subsection{Set Linearizability}\label{sec:set-lin}

Set linearizability~\cite{set-lin94} was proposed by Neiger as an extension to classic linearizability that allows multiple operations to be linearized at the same point.
Operations sharing the same point can be evaluated together non-atomically.
It is subsumed by interval linearizability, but we define it here as another demonstration of X-CCorrectness.

In this case, $f(H)$ is a sequence of sets of operations $O_1, ..., O_n$ such that:
\begin{align*}
    O_1, ..., O_n \in f(H) \text{ iff } & (ops(H) = \bigcup_{x \in 1..n} O_x \\
        & \land \forall x \in \{1..n\} : \forall o, o' \in O_x : \lnot (o <_H o' \lor o' <_H o) \\
        & \land \forall x, y \in \{1..n\} : \\
        & \quad ( O_x \cap O_y = \emptyset \\
        & \quad \land x < y \implies \forall o_x \in O_x, o_y \in O_y : \lnot (o_x <_H o_y)))
\end{align*}

The definition above is equivalent to these rules:
\begin{itemize}
    \item Every operation $o \in ops(H)$ appears in exactly one set $O$.
    \item All operations in the same set must be concurrent (no pairs may be related by $<_H$).
    \item Operations in an earlier set may not happen-after any operation in a later set.
\end{itemize}

$P(O_1, ..., O_n)$ is a set-sequential specification, which might be modeled as a state machine accepting transitions labeled by any set of operations $O$.

\subsection{Sequential Consistency}\label{sec:seq-consistent}

Sequential consistency~\cite{lamport-sc79} (Sequential-CCorrectness by our terminology) is a foundational correctness condition for concurrent objects proposed by Lamport.
Similar to linearizability, operations are mapped to a sequential history, however sequential consistency only requires happens-before ordering between operations with the same process id.
Let $pid(o_i)$ return the process id of operation $o_i$.
$f(H)$ is defined such that:

\begin{equation*}
    o_1, ..., o_n \in f(H) \text{ iff } \forall x, y \in 1..n : (x < y\land pid(o_x)=pid(o_y)) \implies \lnot (o_x >_H o_y)
\end{equation*}

$P(o_1, ..., o_n)$ accepts valid sequences of atomic operations applied to the underlying object.
However, $f(H)$ is more permissive, and fails condition 2 of Theorem~\ref{lem:compose-conditions}.

\begin{proposition}
Sequential consistency is not composable.
\end{proposition}
\begin{proof}
Consider a history of four operations $o_1,o_2,o_3,o_4$ such that $obj(o_1)=obj(o_4)=A$, $obj(o_2)=obj(o_3)=B$, $pid(o_1)=pid(o_2)=1$, $pid(o_3)=pid(o_4)=2$, $o_1<_Ho_2$ and $o_3<_Ho_4$. By definition, $f(H|_A)$ includes the sequence $(o_4, o_1)$ and $f(H|_B)$ includes $(o_3, o_2)$. However, due to the happens-before ordering of $H$, there does not exist a sequence $S\in F(H)$ such that $S|_A=(o_4, o_1)$ and $S|_B=(o_3, o_2)$. By Theorem~\ref{lem:compose-conditions}, $H$ is not composable which implies that sequential consistency is not composable.
\end{proof}

\subsection{Quiescent Consistency}\label{sec:quiescent-consistency}

Quiescent consistency is similar to sequential consistency and classical linearizability in that it requires that a permutation of the operations in the original history satisfy some predicate. However, it only prohibits reordering of operations $o_a, o_b$ if they act on the same object and are separated by a period of ``quiescence'' where the object has no pending operations, i.e. a point where all invoked operations have responded. Formally, define:
\begin{align*}
IsQ(o_a,o_b, H)\text{ iff }&obj(o_a)=obj(o_b)=A\\ &\land\exists o\in H|_A:res(o_a)\leq_{H|_A} res(o) <_{H|_A} inv(o_b)\\ &\quad\quad\quad\quad\quad\land (\lnot \exists o'\in {H|_A}:inv(o')<_{H|_A}  res(o)<_{H|_A}res(o'))
\end{align*}

With $a\leq_Hb$ being shorthand for $a=b\lor a<_H b$. If this predicate is satisfied, then we must have that any sequence generated orders $o_a$ before $o_b$. Thus, $f(H)$ returns all permutations of the operations in the original history such that:
\begin{equation*}
o_1, ..., o_n \in f(H) \text{ iff }\forall x, y \in 1..n : x < y \implies \lnot IsQ(o_y, o_x, H)
\end{equation*}

Then we have that $P(o_1, ..., o_n)$ checks if this sequence satisfies the sequential specification of each object. Using Theorem~\ref{lem:compose-conditions}, we can easily show that quiescent consistency is composable.
\begin{proposition}
Quiescent consistency is composable.
\end{proposition}
\begin{proof}
First note that by definition of $IsQ$, $obj(o_a)=obj(o_b)=A\implies (IsQ(o_a, o_b, H)$ iff $IsQ(o_a, o_b, H|_A))$.
Thus, by definition of quiescent consistency, for any sequence $S$ composed of operations on objects $A$ and $B$, $S|_A \in f(H|_A)\land S|_B\in f(H|_B) \iff S\in f(H)$.
Therefore, we have that $\forall S\in f(H):(\exists S_A\in f(H|_A): S_A\leq_pS|_A)\land(\exists S_B\in f(H|_B): S_B\leq_pS|_B)$ and $\forall S_A\in f(H|_A), S_B\in f(H|_B): \exists S\in f(H): S|_A\leq_pS_A\land S|_B\leq_pS_B$.
By Theorem~\ref{lem:compose-conditions}, this implies that quiescent consistency is composable.
\end{proof}

\end{document}